\documentclass[aps, prd, showpacs, superscriptaddress, groupedaddress, twocolumn]{revtex4} 

\makeatletter
\def\@society{generic}
\makeatother
\usepackage{titlesec}
\usepackage[a4paper, margin=0.6in]{geometry}
\usepackage{graphicx}  
\usepackage{subcaption}
\usepackage{float}
\usepackage{tabularx}
\renewcommand{\thesection}{\arabic{section}}
\usepackage{array}
\usepackage{amssymb}
\usepackage{mathtools}
\usepackage{dcolumn}
\usepackage{color}
\usepackage[font=small,justification=raggedright]{caption}
\titleformat{\section}[hang]
  {\normalfont\large\bfseries}{\thesection}{1em}{}
\usepackage[pagebackref=false, colorlinks=true]{hyperref}
\hypersetup{
linkcolor=blue,
citecolor=blue, 
urlcolor=blue}   
\usepackage{amsmath}
\usepackage{amsthm} 
\usepackage{ulem}
\usepackage[utf8]{inputenc}

\def\be {\begin{equation}}
\def\ee {\end{equation}}
\def\ba {\begin{eqnarray}}
\def\ea {\end{eqnarray}}
\usepackage{caption}
\newtheorem{definition}{Definition}[section]
\usepackage{color}
\def\be {\begin{equation}}
\def\ee {\end{equation}}
\def\ba {\begin{eqnarray}}
\def\ea {\end{eqnarray}}
\usepackage{sidecap}
\newtheorem{prop}{Theorem}
\newtheorem{cor}{Corollary}

\newcommand{\barray}{\begin{array}}
\newcommand{\earray}{\end{array}}

\newcommand{\LB}{\left(}
\newcommand{\RB}{\right)}

\begin{document}
\title{Apparent horizon and causal structure of spacetime singularities}
\author{Koushiki}
\email{koushiki.malda@gmail.com}
\affiliation{International Centre for Space and Cosmology, School of Arts and Sciences, Ahmedabad University, Ahmedabad, GUJ 380009, India}
\author{Pankaj S. Joshi}
\email{psjcosmos@gmail.com}
\affiliation{International Centre for Space and Cosmology, School of Arts and Sciences, Ahmedabad University, Ahmedabad, GUJ 380009, India}
\date{\today}
\begin{abstract} A major issue in general relativity and blackhole physics today is to determine the conditions when the spacetime singularities forming as end-states of gravitational collapse are visible to external observers, and when these are hidden within the event horizon of a black-hole. We show here that such a causal structure of singularity, in terms of its visibility or otherwise, is determined by the dynamics of the apparent horizon and trapped surfaces forming during collapse of massive matter clouds. It turns out that the relative timing of formation of trapped surfaces and the singularity plays a crucial role here. The dynamics of apparent horizon governs the visibility of singularity, and we characterize precise conditions here for spherically symmetric collapse with a general type-I matter field. This is done in terms of the existence of outgoing null geodesic families from the central singularity.  These results hold under generic initial data satisfying $\mathcal{C}^2$ regularity and the weak energy condition.

\end{abstract}
\maketitle

\section{Introduction}
Right after the general relativistic equations of motion were formulated, the theory was presented with a major hurdle in the form of existence of spacetime singularities. Firstly, the Schwarzschild solution, and later the Reissner-Nordström, Kerr and Janis-Newman-Winicour models, were all found to have this feature. Also, the dynamic space-times such as Friedmann-Lemaître-Robertson-Walker, Lemaître-Tolman-Bondi and Vaidya space-times also have singularities. Physical objects can vanish or in principle could emerge from them. The most discussed of them being the big bang singularity in cosmology. And it is not possible to predict the nature and motion of the objects emerging within the current framework of general relativity.\\ 

It turned out that spacetime singularities are inherent geometrical features of these space-times. Penrose, Hawking, and Geroch showed in their singularity theorems \cite{HE}, that space-times, even other than above specific ones, will also possess these  features that emerge from inherent structure of the general relativistic space-time manifold. It is seen that for classes of generic space-times, their maximal Cauchy development allows for families of incomplete non-spacelike geodesics, giving the manifold an edge, which is the singularity, an indispensable, albeit uncomfortable feature of the gravitation theory of Einstein . The space-time manifold structure breaks down at this edge. The visibility, stability and genericity of these singularities are the most puzzling and probably the most important issues here.\\

Among these, the visibility of a singularity has been one of the greatest puzzles in classical gravitation and it continues to be so. Schwarzschild's solution of the Einstein's equations had the first example of a space-time singularity. This is a unique spherically symmetric vacuum solution, going by Birkhoff's theorem \cite{GosEllis}. By virtue of its symmetries, it is one of the simplest examples of a singularity, however, the space-time is completely predictable. These symmetries form a one-way membrane called event horizon, that allows light to enter its internal region: {\it black-hole}, but stops light to come out, thereby making the singularity invisible for all external observers. Following this, it was assumed, by virtue of the {\it cosmic censorship conjecture (CCC)} \cite{CCC} of Penrose, that all space-time singularities must be invisible. But, very few physical situations would possess these symmetries. For example, in physical situations, gravitational objects will have extended mass, very possibly inhomogeneous energy distributions and non-zero pressure profiles. By loosing these strict symmetric conditions of the Schwarzschild configuration, it was found that the generic singularity might not be a black-hole only, but a singularity that is visible. In such situations, it became important to look for local as well as global geometric features of the model space-times.\\

The space-time singularities become specially important 
for the physical problem of gravitational collapse of a  massive star, which at the end of its life runs out of its internal fuel. Since the internal pressures no longer balance its self-gravity, it starts collapsing. If the collapse is violent enough, the end-state is a space-time singularity. We then need to study the structure of the singularity and physics in its neighbourhood. As discussed above, for such a gravitational collapse with inhomogeneous energy distribution and a non-zero pressure profile, the issue of visibility of the final singularity is of crucial importance. This can be examined in terms of escaping future directed null or time-like geodesics which are past incomplete. 
If the singularities of collapse are visible, classical deterministic structure of gravitation theory in the sense of Schwarzschild solution may not exist, and a quantum theory of gravity may describe the physics. In absence of such a theory today, we rely on the CCC \cite{CCC} to preserve predictability in general relativity. Since a proof or 
precise mathematical statement of CCC is not known yet, the only path available is to study different dynamical collapsing models, keeping their description as generic as possible towards using them to describe actual physical systems. This problem was studied by dropping the symmetries one by one and trying to find out the end-state of an {\it unhindered gravitational collapse (UGC)} and its causal or visibility properties. The assumption of homogeneity was dropped by Yodzis, Seifert and Hagen, and an example of a naked singularity was found  \cite{Yodzis:1973gha}. The results by Christodoulou \cite{Christodoulou:1984mz} and Newman \cite{Newman:1985gt} found similar examples and structure of such a singularity was analysed by Newman and Joshi \cite{Newman:1988cr} for inhomogeneous Lemaître-Tolman-Bondi dust clouds. They discussed the differentiability of mass-functions at the beginning of the collapse and their effects on the end-state singularity. For self-similar \cite{Bluman, Strogatz}  clouds, it was shown by Joshi and Dwivedi that there is a class of initial mass functions for which there exists a family of future-directed null geodesics emerging from the singularity \cite{Joshi:1991aa, Joshi:1993zg, Joshi1993, Christ}, thus making naked singularities generic outcomes of UGC. Finally, they omitted the proportionality between the temporal and radial co-ordinates and established the same results for a general inhomogeneous cloud \cite{Dwivedi:1996wf}. These works have taken into account the UGC of spherically symmetric clouds and concretely established that the initial configurations of these clouds, pressure and density profiles alone determine the visibility of the end-state singularity. To show this, nothing but the local structure of the outgoing null geodesics were probed and this showed the existence of families of such geodesics, establishing that the associated singularity is at least locally naked.\\

It was established through these studies that the space-times possessing {\it naked singularities} will not have a globally predictable structure in usual sense like black-holes. So, the definition of the horizon, separating this object from external observers, has to be different from that of the event-horizon. Extending along this line, Ashtekar and Krishnan \cite{Ashtekar:2002} gave a definition of a dynamical {\it apparent horizon (AH)}. This definition does not require information about the asymptotic future of the space-time, making it local, emerging dynamically. None of the works mentioned above, take into account the AH formation or its dynamics, rather they analyze the geodesic equations at the limit of the singularity and this naturally decodes visibility locally. In the context of the formation of a naked singularity, Goswami and Joshi studied the UGC of a collapsing cloud diluting the symmetric conditions further by introducing non-zero pressure and generic type-I matter fields \cite{HE} in the collapsing cloud. AH formation was analyzed here in relation to the visibility of the singularity \cite{Goswami:2002ds, Joshi:2004tb}. Moreover, for a homogeneous perfect fluid with non-zero pressure, it was seen that AH can be completely avoided \cite{Goswami:2004ne, Koushiki:2025dax}. In the same line, the role of shear was investigated in avoiding the trapped surface \cite{Joshi:2001xi}. For a generic type-I matter field the issue of visibility of the singularity is again studied \cite{Goswami:2006ph}, and the formation and evolution of the AH is discussed. The strength \cite{Tipler} of such a singularity was also discussed in this context \cite{Mosani:2020ena, Mosani:2023vtr}. However, the dynamics of the AH is not discussed in these works in relation to the visibility of a singularity in an explicitly analytical manner. In this work, we aim to fill this gap by building an equivalence between the differentiability of the AH curve and the local visibility of the singularity.\\

It is seen in the works involving inhomogeneous perfect fluids, both pressure-less \cite{Goswami:2002he, Goswami:2004gy}, with pressure as well for scalar and vector fields \cite{Mosani:2020ena, Goswami:2007}, that the end-state singularity forms simultaneously as the AH or does not form at all. At this very instant, if a family of outgoing null geodesics manages to escape the singularity, it becomes visible to external observers. Therefore, the local visibility requires the existence of a family of future directed null geodesics whose past lies at the singularity \cite{Joshi:1993zg}. The completeness of such geodesics in the direction of the future decides its global visibility, but it falls out of our current concern. By virtue of the spherical symmetry involved, this problem is studied in the $(t,r)$ plane and the variations in the angular directions does not affect any of its outcomes. The end-state singularity is central and the evolution of the singularity in time is given by the singularity curve and the AH curve. We show that the positivity of the first derivative of the apparent horizon curve in $r$ is equivalent to the the existence of a family of outgoing null geodesics from the singularity. Using the results from previous studies that the existence of a family of outgoing null geodesics from the singularity is the necessary and sufficient condition for the singularity to be at least locally naked, we thus establish that the central singularity is at least locally naked if the AH curve increases in $r$. We extend this result for the higher orders of the expansion of the AH curve and show that the positivity of its lowest non-vanishing order component decides the local visibility of the singularity. This work is done in full generality, only assuming spherical symmetry and type-I matter fields. The initial mass and velocity profiles are also kept generic, only subject to them being $C^2$ functions.

\section{Equations of motion in general spherical symmetry}
The metric of a spherically symmetric space-time for a type-I matter field is expressed by three free functions. In co-moving co-ordinates, its line-element is:
\be \label{gensph}
ds^2 = - e^{2\nu} dt^2 + e^{2\psi} dr^2 + R^2 d\Omega^2,
\ee
where, $\nu \equiv \nu (r,t)$, $\psi \equiv \psi (r,t)$, $R \equiv R (r,t) = a(r,t) r$ and $d\Omega^2$ is the line element of the unit two-sphere. In the units of $8\pi G =c=1$, Einstein's equations of motion, in this case are:
\ba 
\rho &=& \frac{F'}{R^2 R'} \label{E1},\\
p_r &=& -\frac{\dot{F}}{R^2 \dot{R}}\label{E2},\\
\nu' &=&\frac{2(p_\theta -p_r)R'}{(\rho+p_r)R}-\frac{p'_r}{(\rho+p_r)}\label{E3},\\
-2\dot{R}'  + R'\frac{\dot{G}}{G}+\dot{R}\frac{H'}{H} &=&0\label{E4},\\
1-\frac{F}{R} &=& G-H \label{E5}. 
\ea 
This matter field is allowed to have one time-like and three space-like eigen-vectors. In co-moving co-ordinates, the stress-energy tensor is diagonalisable and allowed to have a maximum of four independent components. On top of that, the space-time is also spherical symmetric and because of that, the eigen values associated to the $\theta$ and $\phi$ basis vectors are equal: $p_\theta = p_\phi$. The rest of the functions are defined as:
\ba  
G(r,t) = R'^2 e^{-2\psi} \;&,&\; H (r,t)= \dot{R}^2 e^{-2\nu} \label{GH}\\
F= F(r,t) &=& \frac{R}{2} (1-\nabla_a R \nabla^a R), \nonumber\\
&=& \frac{R}{2} (1-  R'^2 e^{-2\psi} + \dot{R}^2 e^{-2\nu}),\label{Misner}
\ea 
and $F(r,t)$ is the gravitational mass inside a shell of radius $r$ at a time $t$, known as Misner-Sharp mass \cite{Misner:1964}.
From Eq.(\ref{E1}), it can now be seen that $R'= 0$ has to be excluded to preserve regularity at the initial epoch. Also, this configuration has to abide by weak energy configuration. So, the projection of the components of the energy-momentum tensor along any time-like geodesic has to be positive:
\ba
T_{ij}v^i v^j &\geq & 0, \\
\implies \rho \geq 0 \;\; &\text{and}& \rho + p_i \geq 0\;, 
\ea 
and this holds for pressure components along $r^{th}$ and $\theta^{th}$ directions both. Now, using scaling freedom, at the initial epoch of the collapse, we choose:
\be \label{ainit}
a (t_i, r)=1,
\ee
and similarly at the final epoch:
\be \label{afin}
a (t_s (r), r)=0.
\ee
Since, we are discussing a collapsing scenario, the rate of change of the scale-factor is negative:
\be 
\dot{a}<0.
\ee
The collapse configurations in spherical symmetry for type-I scalar field stays the same for any element of this class. We are concerned with this class as most of the physically observable matter fields are covered in this, including homogeneous and inhomogeneous dust, homogeneous and inhomogeneous perfect fluids, also homogeneous and inhomogeneous anisotropic fluids.\\

We do not discuss the dynamics of type-II matter fields here, primarily because it has a double-null eigen-basis for its energy-momentum tensor. Also the only candidate for such a matter field is the null fluid: Vaidya solution. In these same co-ordinate system, its energy-momentum tensor has four independent components, hence they are far more complicated to solve. The other two matter fields \cite{HE} do not obey energy-conditions and also have never been observed in nature. That is why, there is no need to discuss them.

\subsection{Classes of solutions for the collapsing cloud}
It has been noted previously that the initial configuration space decides the end-state of the gravitational collapse \cite{Joshi:1991aa, Goswami:2006ph}. Also, it is important to probe the continuity and differentiability of these functions and ensure that they do not blow up at the initial epoch and also during the collapse, before the formation of the singularity. This is called the regularity condition, which ensures that the initial spatial hypersurface, from which the collapse progresses, is free of a singularity. Other than that, this study becomes irrelevant. So, the dynamics of the collapse depends on the initial configuration of the cloud, specifically the initial configuration space spanned by the functions: $\rho_0 (r), \; p_{r0} (r),\; p_{\theta 0} (r), \;\nu_0 (r)$ and $ \psi_0 (r)$ at the initial epoch $t_i$. As the equations of motion are coupled partial differential equations, some assumptions are to be made to solve these, keeping the generality intact as much as possible. The assumptions to choose these functional forms are listed below:
\begin{itemize}
    \item The mass function is assumed to have the general form:
    \be \label{MF}
F(t,r) = r^3 \mathcal{M} (r,a),
    \ee 
    with $\mathcal{M}>0$ and it being at least $C^1$ for $r\to 0$ and at least $C^2$ for $r>0$. Such a function is chosen so that $\mathcal{M} (r,a)$ resembles the density function in general. 
     \item While the energy density is decided by the mass function, the velocity profile or the pressure of the collapsing cloud is dependent on $\nu$ and to restrict its differentiability, $\nu'$ is assumed to be:
     \be \label{nu'}
\nu' (r,a) = A,_a (r,a) R.
     \ee
     It can be easily seen that $\nu'_0 (r) = A (r,1)$ at $t = t_i$. This is to ensure that the polarity of $R'$ does not change and it must be well-behaved at the central singularity as $t\to t_s$ and $r\to 0$. In this case, $\nu$ can be written as:
     \be \label{nu}
\nu (a,r) = r^2 h_1 (a,r).
     \ee 
     In this way, $\nu (a,r)$ is a $C^2$ function at all regular epochs at the center of the cloud. Now, using Eq.(\ref{nu'}) and Eq.(\ref{nu}), it can be seen that:
     \be \label{A}
A (r,a) = r h_2 (r,a).
     \ee
     The center of the cloud must be regular during all the epochs of evolution before the formation of the singularity. For this requirement, $h_1 (a,r)$ and $h_2 (a,r)$ must be well-behaved at all of these epochs.
     \item There can arise scenarios for a collapsing cloud when $R'<0$. This means that the shells with greater physical radii collapse at a larger speed than the ones with smaller radii, thus making the shells to cross during collapse. If $R'=0$, two collapsing shells can be marked by the same co-moving radius and the co-ordinate system breaks down at these points. This makes the Cauchy extension through these points impossible as the $R\to r$ mapping does not remain one-to-one in such cases causing multiple problems like geodesic break-down and numerical instabilities. Also, from Eq.(\ref{E1}), it can be seen that in such cases the energy density blows up in regular epochs, giving a false singularity. These singularities, therefore, arise from the badness of the co-ordinate system and termed as weak singularities \cite{eardley-smarr1979, hellaby-lake1985, clarke-krolak1985}, not a genuine one, in the sense of Tipler \cite{Tipler}. At the beginning of the collapse, using the scaling freedom, invoked in Eq.(\ref{ainit}), $R' = a + ra' (r,t)$ has the range $1\geq R'>0$.  
     \item Substituting Eq.(\ref{nu'}) in Eq.(\ref{E4}), a constraint on the function $G$ is obtained as:
\be \label{G}
G (r,a) = b(r) e^{2rA (r,A)},
\ee 
where, $b(r)$ is the constant of integration. Physically, it represents the velocity profile of the particles constituting the collapsing cloud, generating the radial pressure. This constraint also ensures the commutativity of the one forms $\frac{\partial}{\partial r}$ and $\frac{\partial}{\partial t}$. $b(r)$ has to be regular at the center of the cloud at all regular epochs making:
\be \label{b}
b(r) = 1+ r^2 b_0 (r).
\ee 
This ensures that radial pressure is finite before the singularity is formed and at the neighbourhood of the singularity it blows up. Using Eq.(\ref{G}) in Eq.(\ref{nu'}) and Eq.(\ref{E3}), it can be seen that Eq.(\ref{b}) is necessary to ensure the regularity of the tangential pressure as well. This same constraint makes the tangential pressure also blow up as the collapse reaches its final singular state.
\end{itemize}
With these general forms at hand, the dynamics of the collapsing cloud can be analysed. Using Eq.(\ref{MF}) Eq.(\ref{G}), Eq.(\ref{b}) in Eq.(\ref{E5}), it can be written:
\ba 
R'^2 e^{-2\psi} - \dot{R}^2 e^{-2\nu}= 1-\frac{\mathcal{M}r^3}{R}\nonumber\\
\implies (1+ r^2 b_0 (r)) e^{2rA}- \dot{R}^2 e^{-2\nu} = 1-\frac{\mathcal{M}r^3}{R}\nonumber\\
\implies \sqrt{R} \dot{R} = -e^\nu \sqrt{(1+ r^2 b_0 (r))R e^{2rA} -R + r^3 \mathcal{M}} \label{Rdot},
\ea 
where, the negative root is only considered to account for the collapsing scenario. To simplify the calculations further, $e^{2rA}$ is written as:
\ba \label{g}
g (r,a) &=& \frac{e^{2rA} -1}{r^2}\\
\implies g (r,a) &=& \frac{2A}{r} + 4A^2 +....\nonumber\\
\implies g (r,a) &=& 2 h_2 (r,a) + \mathcal{O}(r^2), \label{happr}
\ea
ignoring higher order terms and using Eq.(\ref{A}). Using this and Eq.(\ref{Rdot}) it can be seen that:
\be \label{adot}
\dot{a}\sqrt{a} = - \sqrt{ab_0 e^{(rA +\nu)} + e^{2\nu} (ag + \mathcal{M})}.
\ee
This expresses the change of the scale-factor in time, thereby, is the key equation explaining the entire kinematics of the collapse. The initial epoch of the collapse is marked by $a=1$ and any later epoch is marked by $a$, whose value is in the range $0\leq a<1$. The time taken by the configuration to reach from the former to the later, is now calculated as:
\be \label{time}
t(a,r) = \int^1_a \frac{\sqrt{a} da}{\sqrt{ab_0 e^{(rA +\nu)} + e^{2\nu} (ag + \mathcal{M})}}.
\ee 
In this equation, $r$ is a constant and thus $t(a,r)$ can be expanded in Taylor series for small values of $r$ or very close to the center of the cloud:
\be \label{texpansion}
t(a,r) = t (a,0) + r \chi (a) + \mathcal{O} (r^2).
\ee 
The term $\chi$ in the expansion of $t(a,r)$ in Eq.(\ref{texpansion}) is:
\be \label{chi}
\chi (a) \equiv \left(\frac{\partial t (a,r)}{\partial r}\right)_{r=0}.
\ee
In the limit $r\to 0$, using Eq.(\ref{nu}), it can be seen that $\nu$ vanishes. $A$ also vanishes in this limit, and it can be seen using Eq.(\ref{A}). Using these limiting values and replacing Eq.(\ref{time}) in Eq.(\ref{chi}), $\chi (a)$ :
\be \label{CHI1}
\chi (a) = -\frac12 \int_a^1 \frac{\sqrt{a}(a b_{0,r}\vert_{r=0} + a g,_r\vert_{r=0}  + \mathcal{M},_r (r,a) \vert_{r=0})}{{(a b_0\vert_{r=0} + a g\vert_{r=0} + \mathcal{M}(0,a))}^\frac32}~ da.
\ee 
The functions of Eq.(\ref{CHI1}) and their derivatives have to be analysed to comment on the nature of the AH and it is essential to this work. So, for brevity and for ease of calculations, we now define: 
\ba \label{fns}
\mathcal{M}_0 (a) = \mathcal{M}(r,a)\vert_{r=0}~&,&~~ \mathcal{M}_1 (a) = \mathcal{M},_r (r,a) \vert_{r=0}~, \nonumber\\
b_{00} = b_{0}\vert_{r=0}~&,&~~ 
b_1 = b_{0,r}\vert_{r=0}~,\nonumber\\
~~ g_0 = g\vert_{r=0} ~~ &\text{and}& ~~g_1 = g,_r\vert_{r=0}~.
\ea 
With these, Eq.(\ref{CHI1}) becomes:
\be \label{CHI}
\chi (a) =-\frac12 \int_a^1 \sqrt{a} da \frac{ab_1 + ag_1 + \mathcal{M}_1 (a)}{{(ab_{00} + a g_0 + \mathcal{M}_0 (a))}^\frac32}.
\ee 
Eq.(\ref{time}) describes the collapse temporally and hence, the time of formation of the singularity can also be defined by it, which is called the {\it singularity curve}. Dynamically, the singularity is formed when the center of the cloud is described by $zero$ physical radius. Since shell-crossings are ruled out, a singularity in a spherically symmetric collapsing cloud can only form at its center. It can now be seen that the singularity curve is a function of $r$ only. In terms of the scale-factor, the central singularity can be marked by the former vanishing and it is described in Eq.(\ref{afin}). With this, the time of formation of the singularity can be expanded near the center of the cloud as:
\be \label{tsingexp}
t_s (r) = t_{s0} + r \chi (0) + \mathcal{O} (r^2).
\ee
The Taylor expansion works here because the singularity is approached only in its limit. It is clear from this equation that $t_{s0}$ is the time of formation of the central singularity. Using Eq.(\ref{nu}) and Eq.(\ref{A}), as $r\to 0$, this can be calculated as:
\be \label{tsing}
t_{s0} = \int_0^1 \sqrt{\frac{a}{ab_{00} + a g_0 + \mathcal{M}_0 (a)}} da.
\ee
For the central singularity to form in a finite time, $t_{s0}$ has to be well-defined and the denominator cannot go to zero and thus:
\be \label{cond}
ab_{00} + a g_0 + \mathcal{M}_0 (a) >0.
\ee
It is important for the central singularity to form in a finite time, otherwise it is unphysical. Therefore, not worth studying other than a mere mathematical curiosity.  

\section{Apparent horizon and its dynamics} 
In this work, we want to devise a way to comment in a generic manner about the visibility of a collapsing cloud. As mentioned earlier, for a generic collapsing cloud, strict symmetric conditions like homogeneity or the absence of pressure is unrealistic. But completely losing all symmetric conditions make the system of equations impossible to solve analytically. Therefore, this work is done in spherical symmetry and for type-I matter fields. It is important to mention here that spherical symmetry is generic enough for dense collapsing bodies. However, with this assumption, the generation of gravitational waves by the collapsing body is not possible \cite{Berti:2009kk}. But, to move away from spherical symmetry would make the problem almost impossible to be handled analytically. So, we stick to it. Also type-I matter fields \cite{HE} describe all naturally available matter fields, hence very generic. Even with these symmetric considerations, the asymptotic future is not predictable as the space-times become non globally hyperbolic \cite{Joshibook1, Joshibook2}. In this situation, the visibility of the singularity cannot be defined by the event horizon. For such situations, where the global topology becomes indeterminate, a horizon has to be defined locally. Ashtekar and Badrikrishnan \cite{Ashtekar:2003hk} gave such a definition:   
\begin{definition}{\bf Apparent horizon:}
{\it A smooth, three-dimensional sub-manifold $\mathcal{H}$ in a spacetime $(\mathcal{M}, g)$
is said to be a marginally outer trapped surface or an Apparent Horizon (AH) if it is foliated by a preferred family of $2$-spheres such that, on each leaf $\mathcal{S}$, the expansion $\Theta_{l}$ of the outgoing null
normal $k_a$ vanishes and the expansion $\Theta_{n}$  of the ingoing null normal $l_a$ is
strictly negative} \cite{NF}.
\end{definition}\label{mots}
For a spherically symmetric space-time, AH is given by:
\be \label{ah} 
\Theta_l = g^{ij} R,_i R,_j =0,
\ee
where, $R$ is the physical radius of the collapsing cloud and $,_i$ denotes the covariant derivative with respect to the $i$th co-ordinate. Given Eq.(\ref{gensph}) this becomes:
\ba 
R'^2 e^{-2\psi} - \dot{R}^2 e^{-2\nu} = 0,
\ea
which, upon using Eq.(\ref{GH}) and Eq.(\ref{Misner}), locates the AH:
\ba 
 G = H 
\implies \frac{F}{R}=1,
\ea
and $F(r)$ can be clearly identified as the Misner-Sharp mass function \cite{Misner:1964}. Now that the kinematics of the system is well-defined, the dynamics of the visibility can be figured out using the mass and velocity functions. For that, the mass-function is determined at the AH using Eq.(\ref{MF}) and the expression of the scale-factor :
\be \label{mots}
\mathcal{M} (r,a) r^2 =a.
\ee 
In the previous section, it was discussed that $\mathcal{M} (0,a) \equiv \mathcal{M}_0 (a)$ is a well-behaved function at the center of the cloud. And Eq.(\ref{mots}) gives the curve of the AH. Using this and Eq.(\ref{time}), the time of formation of the AH can be calculated as:
\be \label{timeAH}
t_{AH}(r) = \int^1_{a_{AH}(r)} \frac{\sqrt{a} da}{\sqrt{ab_0 e^{(rA +\nu)} + e^{2\nu} (ag + \mathcal{M})}}.
\ee 
This is called the {\it AH curve}, which is just like the singularity curve, a function of only $r$. Therefore, with $\mathcal{M}_0 (a)$ well defined, the AH curve can also be expanded close to $r\to 0$, in a similar way of Eq.(\ref{texpansion}) and can be written as:
\be \label{Ahexp}
t_{AH} (r) = t_{AH0} + r \Psi (0) + \mathcal{O} (r^2).
\ee 
The general time-curve of Eq.(\ref{time}) can be integrated from the start of the collapse $(a=1)$ to the occurrence of the singularity $(a=0)$, via the time of formation of the AH $(a=a_{AH})$. This gives the singularity curve as:
\ba \label{sinah}
t_s(r) = &-&\int_1^{a_{AH} (r)} \frac{\sqrt{a} da}{\sqrt{ab_0 e^{(rA +\nu)} + e^{2\nu} (ag + \mathcal{M})}}\nonumber \\
&-& \int_{a_{AH} (r)}^0 \frac{\sqrt{a} da}{\sqrt{ab_0 e^{(rA +\nu)} + e^{2\nu} (ag + \mathcal{M})}},
\ea 
where, $a_{AH} (r)$ is given by Eq.(\ref{mots}). Comparing this with Eq.(\ref{timeAH}), it can be seen that the first term in this equation represents the time of formation of the AH. Thus, this equation becomes:
\be \label{TSIN}
t_s(r) =  - \int_{a_{AH} (r)}^0 \frac{\sqrt{a} da}{\sqrt{ab_0 e^{(rA +\nu)} + e^2\nu (ag + \mathcal{M})}} +t_{AH} (r).
\ee 
With the singularity and the AH curve at our disposal, we move to the main goal of this work: to establish a one-to-one relationship between the singularity and AH curve and the local visbility of the singularity for a generic type-I collapsing cloud in spherical symmetry. We show that in the three following steps:
\begin{itemize}
    \item For generic clouds, the singularity and AH forms at the same instant. Their slopes at the center are identical as well. 
    \item In this case, the higher order term, or the slope of these curves, have to be studied to comment on the local visibility. 
    \item The sign of the lowest order non-vanishing term of the AH curve decides the local visibility of the singularity.
\end{itemize}
With this, we present our first proposition:
\begin{prop}\label{theorem1}
    For a generic inhomogeneous initial data, the times of formation of the central singularity and the AH are the same at the end of a UGC. Moreover, the expansion of their curves in the neighbourhood of the singularity are the same. 
\end{prop}

\begin{proof}
To compare the singularity curve with the AH curve at the center of the cloud, we take $r\to 0$ in Eq.(\ref{TSIN}):
\ba \label{ts=tah}
\lim_{r\to 0} t_s(r) &=&   -\lim_{r\to 0}  \int_{a_{AH} (r)}^0 \frac{\sqrt{a} da}{\sqrt{ab_0 e^{(rA +\nu)} + e^{2\nu} (ag + \mathcal{M})}}  \nonumber \\ &+& \lim_{r\to 0}  t_{AH} (r).
\ea 
In this limit, using Eq.(\ref{mots}), it can be seen that $a_{AH} (r)\to 0$.  As $r\to 0$, the numerator of the first term vanishes. In this limit, $\nu \to 0$, using Eq.(\ref{nu}) and $A \to 0$ using Eq.(\ref{A}). Using these limiting values of $\nu$ and $A$ and Eq.(\ref{cond}), it can also be seen that the denominator of this term is finite. So, this term vanishes in the limit $r\to 0$ and we get:
\ba \label{eql} 
\lim_{r\to 0} t_s(r) &=& \lim_{r\to 0}  t_{AH} (r)\nonumber\\ 
\implies t_{s0} & =& t_{AH0}. \label{eql}
\ea 
It is clear now that the central singularity and the AH forms at the same instant. Therefore, the situation is more ambiguous than that of a black-hole, where the AH forms before the formation of the singularity with the trapped surfaces completely covering it. So, it stops outgoing null congruence from even forming near it, making it necessarily hidden. But, here the situation is different as the existence of an outgoing null congruence cannot be ruled out based only on the fact that the epoch of formation of the singularity and the trapped surfaces coincide. 
Using the result of Eq.(\ref{eql}) in Eq.(\ref{Ahexp}), it becomes:
\be \label{motsexp}
t_{AH} (r) = t_{s0} + r \Psi (0) + \mathcal{O} (r^2).
\ee 
Since the zero-th order term in the expansion of $t_{s0}$ and $t_{AH0}$ cannot resolve the matter of visibility of the central singularity in this case, the next higher order term has to be looked upon to comment on the same. For this, Eq.(\ref{tsingexp}) has to be compared to Eq.(\ref{motsexp}). By doing so, it becomes evident that $\chi (0)$ and $\Psi (0)$ both have the same general form given in Eq.(\ref{CHI}). For $\chi (0)$, the lower limit of the integration is clearly $a \to 0$, as given in Eq.(\ref{afin}). Whereas, $\Psi (0)$ is given by taking the lower limit of the integration to $a_{AH} (r)$ as $r\to 0$. With Eq.(\ref{mots}), it can be seen that in this limit $a_{AH} = 0$. Therefore, it can be said that:
\be \label{WWW}
\lim_{r\to 0} \frac{dt_s (r)}{dr} = \chi (0) = \lim_{r\to 0} \frac{dt_{AH} (r)}{dr}. 
\ee 
{\it If the singularity curve is increasing at the center of the collapsing cloud, then the AH curve must also be increasing and vice-versa.} 
\end{proof}
However, this still leaves the question whether this is sufficient for the central singularity to be locally naked.
\subsection{The necessary and sufficient condition for a singularity to be locally naked} \label{result}

The concept of a naked singularity is relatively very new in classical relativity. It is defined as \cite{Joshi:1991aa}:
\begin{definition}{\bf Naked singularity:}
A singularity is naked if there are future directed non space-like geodesics in the space-time
with their past end-point in the neighbourhood of the singularity. If the causal future of such geodesics are incomplete, then the singularity is locally naked. If these geodescis are complete in their future, then the singularity is globally naked. 
\end{definition}
So, a singularity is certainly locally naked if $\exists$ a family of causal geodesics emerging out of the singularity. We focus our investigation to radially outgoing geodesics only, as this significantly reduces the complications of the analysis. Also, existence of null geodesics make room for the existence of time-like geodesics. And in spherical symmetry, analysing the radial geodesics is enough, as the rest do not bring out new Physics. Now, the null geodesic equation must have a root as $t\to t_s$ and $r\to 0$. And for the geodesic to be outgoing the tangent to it must be positive. Hence, the root to the geodesic equation must also be positive:
\be 
\frac{dt}{dr} = e^{\psi -\nu}.\label{RNG}
\ee 
This is derived from Eq.(\ref{gensph}) with $d\theta = d\phi =0$, without disturbing the generality of the solution. By transforming $u\to r^\alpha$ with $\alpha >1$, this equation becomes:
\be 
\frac{dR}{du} =\frac{R'}{\alpha r^{\alpha -1}} \left[ 1+\frac{\dot{R}}{R'}e^{\psi-\nu}\right].\label{RNG1}
\ee 
From this equation, it can be seen that with $\alpha <1, \;\; \frac{R'}{\alpha r^{\alpha -1}}$ blows up with no shell-crossings, making $\frac{dR}{du}$ blow up as well in the limit $r\to 0$. So, we constraint it: $\alpha > 1$. Now, using Eq.(\ref{E5}) and Eq.(\ref{GH}) and $\alpha = \frac53$, this equation becomes:
\be \label{dRdu}
\frac{dR}{du} =\frac35 \LB \frac{R}{u} + \frac{a'\sqrt{a}}{\sqrt{\frac{R}{u}}} \RB \LB \frac{1-\frac{F}{R}}{\sqrt{G} (\sqrt{G}+\sqrt{H})}\RB.
\ee 
This equation is of the outgoing null geodesic in the $(R,u)$ plane. Near the singularity, using L'H\^{o}pital's rule, this becomes:
\be\label{ORNG}
x_0 \equiv \lim_{\substack{t\to t_s\\ r\to0}} \frac{R}{u} = \lim_{\substack{t\to t_s\\ r\to0}} \frac{dR}{du}.
\ee
It can be clearly seen that near the singularity, $x_0$ is the tangent to the outgoing radial null geodesic. So, $x_0$ is the root of Eq.(\ref{dRdu}) near the singularity. To evaluate its value, the value of $a'\sqrt{a}$ also has to be known. Applying Leibniz rule on Eq.(\ref{adot}):
\ba
\sqrt{a}~\frac{\partial a}{\partial r} ~\frac{\partial r}{\partial t}= - \sqrt{ab_0 e^{(rA +\nu)} + e^{2\nu} (ag + \mathcal{M})}.
\ea 
For $r\to 0$ and for constant $a$ surfaces, the bi-variate function $t(a,r)$ can be expanded as in Eq.(\ref{texpansion}). Hence, in this limit $r\to 0$, this equation becomes:
\ba
a' ~\sqrt{a} &=& \sqrt{ab_{00}+ a g_0 + \mathcal{M}_0 (a)}~ \chi (a) + \mathcal{O} (r) \label{aprime}.
\ea 
To evaluate the root of Eq.(\ref{dRdu}), $\frac{dR}{du}$ has to be found out in the limit $(t,r)\to (t_s,0)$, which has earlier been named as $x_0$ in Eq.(\ref{ORNG}). Its positivity ensures the existence of a family of outgoing null geodesics from the singularity \cite{Dwivedi:1996wf}. And to evaluate $x_0$, the metric functions have to be evaluated as $(t,r)\to (t_s,0)$:
\begin{itemize}
    \item Substituting Eq.(\ref{A}) and Eq.(\ref{b}) in Eq.(\ref{G}), it can be seen that $G(t_{s0},0)\to 1$.
    \item From Eq.(\ref{nu}), it can be seen that $\nu \to 0$ in this limit. Also, from Eq.(\ref{GH}), $H$ has the value $e^{-2\nu} \dot{R}^2$. Therefore, in this limit $H\to \dot{R}^2$. From Eq.(\ref{A}), the limiting value of $A$ can be found out be $0$. And Eq.(\ref{g}) gives  $g\to 0$. Now, using Eq.(\ref{adot}), it can be seen that $H\approx \frac{r^2}{a}$. So, in the limit $(t,r)\to (t_{s0},0)$, $H\to 0$. 
    \item Using the values of $G$ and $H$ in the limit $(t,r)\to (t_{s0},0)$, obtained above, in Eq.(\ref{E5}), the limiting value of $\frac{F}{R}$  can be seen to be vanishing. 
\end{itemize}
With these, now we reach our second theorem:
\begin{prop}\label{theorem2}
The positivity of $\chi_0$ is the necessary and sufficient condition for $x_0$ to be positive, thus making $\chi_0 >0$ the necessary and sufficient condition for the central singularity to be at least locally naked.
\end{prop}

\begin{proof}
Substituting the limiting values of the functions of $G$, $H$ and $\frac{F}{R}$ as $(t,r)\to (t_s,0)$, as discussed above in Eq.(\ref{dRdu}), we get:
\be 
x_0 = \frac35 \left( x_0 + \frac{a' \sqrt{a}}{\sqrt{x_0}}\right).
\ee  
Now, we substitute $a' \sqrt{a}$ from Eq.(\ref{aprime}) in this equation. The higher order terms in $r$ vanishes in this limit and hence, it can be seen that:
\ba
    x_0 &=& \frac{3}{5} \LB x_0 + \frac{\sqrt{\mathcal{M}_0} ~\chi (0)}{\sqrt{x_0}}\RB \nonumber\\
\label{X} \implies {x_0} &=& {\LB \frac{3}{2} \sqrt{\mathcal{M}_0} \chi (0)\RB}^\frac23
\ea
Eq.(\ref{X}) is the equation of the outgoing radial null geodesics in the $(R,u)$ plane. And $x_0$ being positive necessarily means that $\chi_0$ is positive as well, and vice-versa. Shifting from $(t,r)$ co-ordinates to $(R,u)$ does not affect the causal structure or visibility criterion. So, these results are valid for all good co-ordinate systems. Now, by replacing $u=r^\alpha$, this curve can be retrieved in the $(t,r)$ plane:
\be \label{singcurve}
t-t_s (0) = x_0 r^\frac53.
\ee 
From this equation, it is clear that $x_0$ is a solution to the null geodesic equation at the central singularity. Also,  if and only if $x_0>0$, then with increasing $r$, $t$ also increases. In other words with a positive $x_0$, there is a family of outgoing radial null geodesics emerging out of the central singularity. \textit{So, $\chi (0)>0$ is the necessary and sufficient condition for the singularity to be at least locally naked.}
\end{proof}

\begin{cor}
Theorem \ref{theorem2} also shows that an increasing AH necessarily implies that the singularity curve is also increasing at the center. Therefore, an increasing AH curve or an increasing singularity curve at the center are the necessary and sufficient conditions for the singularity to be at least locally naked.
\end{cor}

\begin{proof}
It can be seen from Eq.(\ref{X}), that $x_0 >0 \iff \chi (0)>0$ as $\mathcal{M}_0$ is a well-defined positive constant. Now, using Eq.(\ref{WWW}), it can also be seen that $x_0 >0 \iff \lim_{r\to 0} \frac{dt_s (r)}{dr} = \lim_{r\to 0} \frac{dt_{AH} (r)}{dr} >0$. Therefore, proving the positivity of the of the tangent of either the AH or the singularity curve is enough to show that the central singularity is necessarily locally naked.
\end{proof}

\subsection{Higher order derivatives of the AH curve and its relationship to the local visibility of the singularity}
If the first order derivative of the AH and the singularity curve vanish, then it is necessary to look at the higher order derivatives of the same. This becomes important in the studies where a similar approach is needed, in probing the stability or the genericity of the resultant singularity \cite{Satin:2014ima}. In these studies, the first order derivative might vanish or even the higher orders, depending on the order of perturbation considered in these problems. For such cases, it is important to look at these higher order derivatives of the AH. With this, we reach our third and final theorem:
\begin{prop}\label{theorem3}
If $\chi (0)$ vanishes for any physical scenario, the local visibility of the central singularity at the end of a UGC is ensured by the positivity of the next non-zero higher order term in $r$ in the expansion of the AH curve.
\end{prop}

\begin{proof}
Following from Eq.(\ref{texpansion}), the time-curve for the collapse can be expanded with all higher order terms up to order $n$:  
\ba \label{texpansionHO}
t(a,r) &=& t (a,0) + r \chi (a)+ r^2 \chi_2 (a) + r^3 \chi_3 (a)\nonumber \\ &+&.....+ r^n \chi_n (a) + \mathcal{O} (r^{n+1}),
\ea
where, 
\be 
\chi_n (a) \equiv \frac{1}{n!} \left(\frac{\partial^n t (a,r)}{\partial r^n}\right)_{r=0}.
\ee
In the following, we will discuss the case where $\chi (0)=0$ and $\chi_2 (0)\neq 0$ and derive the necessary and sufficient condition for the central singularity to be at least locally naked. By differentiating Eq.(\ref{CHI}), $\chi_2$ can be calculated as:
\ba 
\chi_2 (a) = \frac{3}{8} \int_a^1 \sqrt{a} da   \frac{(ab_1 + ag_1 +\mathcal{M}_1 (a))^2}{(ab_{00}+ag_0 + \mathcal{M}_0 (a))^\frac52}\nonumber \\
-\frac{1}{4} \int_a^1 \sqrt{a} da \frac{(ab_2 + ag_2 +\mathcal{M}_2(a))}{{(ab_{00}+ag_0 + \mathcal{M}_0 (a))}^\frac32},
\ea 
where, $b_2 = b_{0,rr}\vert_{r=0}$, $\mathcal{M}_2 (a) = \mathcal{M},_{rr} (r,a) \vert_{r=0}$ and $g_2 = g,_{rr}\vert_{r=0}$ following the convention used in Eq.(\ref{fns}). Now, going back to Eq.(\ref{RNG1}) and keeping the value of $\alpha$ arbitrary, this equation can be rewritten as:
\be \label{dRdu1}
\frac{dR}{du} =\frac{1}{\alpha} \LB \frac{R}{u} + \frac{a'\sqrt{a}r^{\frac{5-3\alpha}{2}}}{\sqrt{\frac{R}{u}}} \RB \LB \frac{1-\frac{F}{R}}{\sqrt{G} (\sqrt{G}+\sqrt{H})}\RB.
\ee 
Using Eq.(\ref{aprime}) and repeating the method to derive it, the value for $a'\sqrt{a}$ can be calculated if the higher order derivatives of the time-curve does not vanish:
\ba \label{aprime1}
a' ~\sqrt{a} &=& \sqrt{ab_{00}+ a g_0 + \mathcal{M}_0 (a)}~ (\chi (a)+ 2r \chi_2 (a) \nonumber \\ &+& 3r^2 \chi_3 (a)+.....+ nr^{n-1} \chi_n (a) + \mathcal{O} (r^{n}) ).
\ea
In the limit of $r\to 0$ and $a\to 0$ this equation becomes: 
\ba \label{aPRIME}
 \lim_{\substack{t\to t_s\\ r\to0}} a'\sqrt{a} &=& \sqrt{ \mathcal{M}_0 (0)}~(\chi (0)+ 2r \chi_2 (0)+ 3r^2 \chi_3 (0)\nonumber \\ &+& ....+ nr^{n-1} \chi_n (0) + \mathcal{O} (r^{n}) ),
\ea
using the limiting values of the functions calculated in the discussion following Eq.(\ref{aprime}). Repeating the calculations preceding Eq.(\ref{X}) and using Eq.(\ref{dRdu1}) and Eq.(\ref{aPRIME}), the root equation for the outgoing null geodesic congruence is arrived at:
\ba\label{high}
{x_0}^\frac{3}{2} &=& \lim_{r\to 0} \frac{1}{\alpha -1} (\chi (0)+ 2r \chi_2 (0)+ 3r^2 \chi_3 (0) +....\nonumber \\ &+& nr^{n-1} \chi_n (0) + \mathcal{O} (r^{n}) ) r^{\frac{5-3\alpha}{2}} \sqrt{ \mathcal{M}_0 (0)}. 
\ea 
From this equation, it is evident that if $\chi (0)=0$ and $\chi_2 (0)\neq 0$, the second-order behaviour cannot be correlated to the local visibility of the singularity if $\alpha = \frac53$. But if $\alpha = \frac73$, then the preceding equation becomes:
\be \label{HO1}
{x_0} = {\LB \frac{3}{2} \sqrt{\mathcal{M}_0} \chi_2 (0)\RB}^\frac23,  
\ee
which is just the same in its functional form as Eq.(\ref{X}). From Eq.(\ref{high}) it can be seen, if $\chi(0)=\chi_2(0)=0$ but $\chi_3(0)\neq 0$, then $\alpha = 3$ and in this case, the root equation becomes:
\be \label{HO2}
{x_0} = {\LB \frac{3}{2} \sqrt{\mathcal{M}_0} \chi_3 (0)\RB}^\frac23. 
\ee 
From Eq.(\ref{HO1}) and Eq.(\ref{HO2}), it is clear that the root equation of the outgoing null geodesic congruence gives a positive solution if the lowest non-zero derivative of the AH curve is positive. This result holds for all higher order derivatives as well for:
\be \label{alpha}
\alpha = \frac{2n+1}{3};\{n \in \mathbb{N}\;\vert \; n\geq 2\}.
\ee 
Therefore, the result of theorem \ref{theorem2} holds for the first non-vanishing derivative in the expansion of the AH curve, near the central singularity, given the mass function and the velocity distribution functions of the collapsing cloud follows appropriate continuity and differentiability conditions. And in this neighbourhood, it is enough to just look at the positivity of the first non-zero term of the expansion of the AH curve to comment on the local visibility of the singularity.\end{proof}
Continuity and differentiability of these functions ensure that $\chi_n (0)$ does not blow up. For this, the lowest order of $\chi_n (a)$ has to be positive in $r$ in the limit $(a,r)\to (0,0)$:
\ba
\chi_n (a ) &\sim & \mathcal{O} \left(r^{\frac{3\alpha -1}{2}-j}\right),\\
\text{with}\;\;\;  \frac{3\alpha -1}{2} &\geq & j.
\ea 
Given all these conditions follow through, the positivity of the AH curve or the singularity curve expansions in its lowest non-vanishing order is the necessary and sufficient condition for the local visibility of an end-state singularity.

\section{Comments on the causal nature of the central singularity}
This discussion remains incomplete without commenting on the nature of the central singularity, whether it is time-like, null or space-like. For that, the line element (with $d\theta = d\phi =0$) of Eq.(\ref{gensph}) has to be studied in the limit $(t,r)\to (t_{s0},0)$. Ignoring the angular contribution does not affect the generality of the discussion by virtue of the spherical symmetry and it is far simpler in terms of calculation. Along the singularity curve:
\ba \label{ts0}
t_s - t_s (0) = {\LB \frac32 \sqrt{\mathcal{M}_0} \chi (0)\RB}^\frac23 r^\frac53,
\ea 
replacing Eq.(\ref{X}) in Eq.(\ref{singcurve}).
Using this, the line-element $-g_{tt}dt^2 + +g_{rr}dr^2$ is evaluated in the limit $(t,r)\to (t_{s0},0)$. Employing Eq.(\ref{GH}), Eq.(\ref{G}) and Eq.(\ref{b}):
\ba 
-g_{tt}dt^2 + +g_{rr}dr^2 &=& -e^{2\nu} {\LB\frac{dt}{dr} \RB}^2 dr^2 + e^{-2rA} dr^2\nonumber\\
&=& \frac53 r^\frac23 {\LB\frac32 \sqrt{\mathcal{M}_0 }\chi (0) \RB}^\frac43 (-dr^2 + dr^2)\nonumber\\
&=&0.
\ea 
So, along the singularity curve, the line element is null, making the singularity null as well. Using theorem \ref{theorem1}, it can now be safely said that if this curve is increasing, then the AH curve is increasing as well. And the positivity of its first non-vanishing order ensures the positivity of the first non-vanishing order of the AH curve. In such a case, using theorems \ref{theorem2} and \ref{theorem3}, it can be established that this singularity is at least locally naked.

\section{Discussion}
In this work, we have commented on the nature of the AH and its relationship with the local visibility of the central singularity at the end of a collapse. This topic is more delicate compared to the global visibility of a singularity, which is formed by the complete avoidance of trapped regions and the formation of an AH. Therefore, this treatment can be of help trying to study the visibility of singularities in similar set-ups. Especially for figuring out critical gravitational collapse \cite{Choptuik, Abrahams, Hamad:2004sw, Garfinkle1998, Brady1995}, where the phenomena is only known, not the exact cause behind it. Also, our treatment is fairly generic, with the only assumptions being the matter being of type-I and the cloud being spherically symmetric. Despite this, there are few nuances that we would like to point towards now for future works:
\begin{itemize}
\item Just the analytical modeling of a locally naked singularity is not enough, if the detection of such objects is impossible. Lately, {\it Gravitational Wave (GW)s} have become one important way to do the same \cite{Bailes:2021tot}. However, generation of GWs is not possible in spherical symmetry, not without some aspherical perturbation or other perturbations \cite{Betschart:2004uu, Zibin:2008vj}. Naked singularities have been observed in non-spherical collapsing objects as well \cite{Shapiro:1992heg, DeLaCruz:1970kk, Price:1971fb, Joshi:1996qc}. Therefore, a natural curiosity will be there to find whether the visibility of a singularity can be related to the AH, even in the cases of non-spherical collapses.
\item Another possible way to detect singularities could be through GW observations of merger events \cite{LIGOScientific:2024elc}. However, such merger events are correlated to theoretical predictions through sophisticated numerical simulations that need global slicings of the space-times \cite{LIGOScientific:2019lyx}. These methods, in their current form, cannot be applied to a space-time containing a naked singularity, as the space-time is not globally hyperbolic. It would also be important and intriguing to resolve this issue, either by developing numerical techniques that work with local slicings \cite{EllisLRS, Clarkson:2002jz} or developing a new analytical methodology that is fit for numerics. 
\item Another possible way to observe such singularities is by studying the electro-magnetic radiation emitted by objects circling or orbiting the singularities and producing luminosity \cite{Kong:2013daa, Joshi:2013dva}. It is not hard to understand that these rays will travel through other space-times, other than the one described here, before reaching the observers. Such space-times have to be matched with the space-time described here to comment on the actual nature of these rays. This has been done to analytically comment on such light rays \cite{Goswami:2006ph, Mosani:2023vtr, Koushiki:2025dax}. These has to be verified by observational data to comment concretely on the validity of CCC.
\item Another debate surrounding UGC is the genericity \cite{psj} and stability \cite{senovilla, Christ1} of the end-states produced out of them. Despite the absence of a rigorous definition for any of these, these are often mentioned as the reasons for discarding the existence of naked singularities which are not the same as the black-holes. This is a serious concern and interesting problems can be developed if these issues are seriously studied. To probe these issues, one would need to understand the topologies of the function space constituting the initial configuration of the UGC. Then by defining a measure on this space, definitions of genericity and stability can be argued upon. In this work, the function space comprising of $F$, $G$, $\nu$, $\mathcal{M}$, is fairly general and can be treated as the starting point of such a study.
\item This study only takes into account the collapse of a type-I matter cloud, but there are also type-II matter fields. The UGC of these matter fields can also be studied in a similar manner to put an overarching comment on CCC or its non-existence. It is to be noted that this matter field has also been studied in a fairly general manner and one possible end-state of such a collapse has been shown to be a naked singularity \cite{Ghosh:2002mf, Brassel:2021}. It would be interesting to find a one to one correspondence between the dynamics of the AH and the visibility of the singularity in this case as well. This work can be a starting point of such a study as well.
\end{itemize}
To summarise, this work points out   that the visibility of the singularity can be commented upon just by studying the co-efficients of different non-vanishing orders of the AH near it. With these results that uses an analytic framework derived with a fairly general configuration space, the case of a naked singularity must be regarded as an equally probable outcome of a UGC, just like the black-hole. Just like general relativity gives counter-intuitive results in black-holes, a naked singularity can also be regarded as a similar counter-intuitive result. Moreover, focus should be put upon figuring out the physical reasons behind the delaying or non-occurrence of the AH. The reasons can be many, including non-zero pressure \cite{Mosani:2023vtr, Goswami:2002ds} or non-vanishing shear \cite{Joshi:2001xi, Joshi:2004tb}. Using the results of this work, the visibility of a singularity might be investigated and based on that the physical reasons behind that can be probed in future.

\end{document}